%% file: main.tex
\documentclass[runningheads, envcountsect]{llncs}
\input{preamble}
\usepackage{graphicx}
\usepackage[utf8]{inputenc} 
\usepackage[hidelinks]{hyperref}       
\usepackage{url}            
\usepackage{booktabs}       
\usepackage{amsfonts}       
\usepackage{nicefrac}       
\usepackage{microtype}      

\newcommand{\agnote}[1]{\refstepcounter{note}$\ll${\bf Anupam~\thenote:}
  {\sf \color{purple}  #1}$\gg$\marginpar{\tiny\bf AG~\thenote}}

\newcommand{\price}{\ensuremath{\pi}}
\newcommand{\bprice}{\boldsymbol{\price}}
\newcommand{\probed}{\mathsf{Probed}}
\newcommand{\val}{\mathsf{val}}
\newcommand{\cost}{\mathsf{cost}}
\newcommand{\I}{\mathbb{I}}

\newcommand{\MS}{\mathcal{S}}
\newcommand{\X}{\mathbf{X}}
\newcommand{\cali}{\mathcal{I}}
\newcommand{\Y}{\mathbf{Y}}

\newcommand{\R}{\mathbb{R}}

\newcommand{\A}{\ensuremath{\mathcal{A}}\xspace}

\newcommand{\frugal}{\text{\sc{Frugal}}\xspace}
\newcommand{\robust}{\text{\sc{Robust}}\xspace}

\newcommand{\dest}{{d}\xspace}
\newcommand{\MPOI}{{\sc Markovian PoI}\xspace}
\newcommand{\reward}{r\xspace}
\newcommand{\breward}{{\bf \reward}\xspace}
\newcommand{\prepared}{{ready}\xspace}
\newcommand{\FreeInfo}{{\sc Free-Info}\xspace}
\newcommand{\traj}{\omega}
\newcommand{\trajProf}{{\boldsymbol{\omega}}}
\newcommand{\taumax}{\ensuremath{\tau^{\max}}\xspace}

\newcommand{\Ym}{\ensuremath{\Y_{M}}\xspace}
\newcommand{\YmH}{\ensuremath{\widehat{\Y}_{M}}\xspace}
\newcommand{\Ymmax}{\ensuremath{\Ym^{\max}}\xspace}

\newcommand{\YmHmin}{\ensuremath{\YmH^{\min}}\xspace}
\newcommand{\YmHmax}{\ensuremath{\YmH^{\max}}\xspace}
\newcommand{\bYmin}{\ensuremath{\mathbf{Y^{\min}}}\xspace}
\newcommand{\bYHmin}{\ensuremath{{\widehat{\mathbf{Y}}^{\min}}}\xspace}
\newcommand{\bYHmax}{\ensuremath{{\widehat{\mathbf{Y}}^{\max}}}\xspace}
\newcommand{\bYmax}{\ensuremath{\mathbf{Y^{\max}}}\xspace}
\newcommand{\Ymin}{\ensuremath{Y^{\min}}\xspace}
\newcommand{\YHmin}{\ensuremath{\widehat{Y}^{\min}}\xspace}
\newcommand{\Ymax}{\ensuremath{Y^{\max}}\xspace}
\newcommand{\YHmax}{\ensuremath{\widehat{Y}^{\max}}\xspace}
\newcommand{\Rmin}{\ensuremath{R^{\min}}\xspace}

\newcommand{\eat}[1]{}

\newcommand{\ALG}{\ensuremath{\text{\sc{ALG}}}\xspace}
\newcommand{\ALGA}{\ensuremath{\text{\sc{ALG}}_\A}\xspace}

\newcommand{\ALGH}{\ensuremath{\widehat{\text{\sc{ALG}}}}\xspace}
\newcommand{\ALGHA}{\ensuremath{\widehat{\text{\sc{ALG}}}_\A}\xspace}
\newcommand{\TG}{\ensuremath{G_T}\xspace}

\newcommand{\Det}{\textsc{Deterministic}\xspace}
\newcommand{\UtiMax}{{\sc Util-Max}\xspace}
\newcommand{\DisMin}{{\sc Disutil-Min}\xspace}
\newcommand{\Umin}{\ensuremath{U^{\min}}\xspace}
\newcommand{\Umax}{\ensuremath{U^{\max}}\xspace}

\newcommand{\MCH}{\widehat{\mathcal{S}}}

\newcommand{\pH}{\widehat{p}}
\newcommand{\PH}{\widehat{P}}
\newcommand{\tauH}{\widehat{\tau}}
\newcommand{\gamH}{\widehat{\gamma}}
\newcommand{\GH}{\widehat{G}}
\newcommand{\SC}{\mathcal{S}}

\newcommand{\Sur}{\ensuremath{\text{\sc SUR}}\xspace}
\newcommand{\Rob}{\ensuremath{\text{\sc Robustness}}\xspace}
\newcommand{\DAG}{\text{\sc{Dag}}\xspace}
\newcommand{\pol}{\ensuremath{\text{\sc Pol}}\xspace}

\newcommand{\ctr}{\ensuremath{\text{\rm{ctr}}}}
\newcommand{\etal}{et al.}
\newcommand{\poly}{\ensuremath{\text{\rm{poly}}}}

\newenvironment{proofof}[1]{\smallskip\noindent{\bf Proof of #1.}}%
        {\hspace*{\fill}$\Box$\par}

\renewcommand{\backslash}{\setminus}

\begin{document}
\title{The Markovian Price of Information}
%
%
\author{Anupam Gupta \inst{1} \and
Haotian Jiang\inst{2} \and
Ziv Scully\inst{1} \and
Sahil Singla\inst{3}}
\authorrunning{Gupta et al.}
%
\institute{Carnegie Mellon University, Pittsburgh PA 15213, USA \and
University of Washington, Seattle WA 98195, USA \and
 Princeton University, Princeton NJ 08544, USA 
}
\maketitle              
\begin{abstract}
\input{abstract}
\end{abstract}

\input{introIPCO}







 \bibliographystyle{splncs04}
 \bibliography{bib}

\newpage
\input{appendix}

\end{document}

%% file: preamble.tex
\usepackage{subfig}
	\usepackage{latexsym,graphicx,epsfig,color}
\usepackage{mathtools}
\usepackage{url,setspace}
\usepackage{multirow}
\usepackage{rotating}
\usepackage{tikz}
\usepackage[shortlabels]{enumitem}
\usepackage{xspace}
\usepackage{algorithm,algpseudocode}
\usepackage{algpseudocode}
\usepackage{bm}

\newcommand{\IGNORE}[1]{}

\usetikzlibrary{decorations.markings}
\usetikzlibrary{arrows}
\tikzstyle{block}=[draw opacity=0.7,line width=1.4cm]
\tikzstyle{graphnode}=[circle, draw, fill=black!20, inner sep=0pt, minimum width=6pt]
\tikzstyle{point}=[circle, draw, fill=black!30, inner sep=0pt, minimum width=1pt]
\tikzstyle{input}=[rectangle, draw, fill=black!75,inner sep=3pt, inner ysep=3pt, minimum width=4pt]
\tikzstyle{unmatched}=[graphnode,fill=black!0]
\tikzstyle{shaded}=[graphnode,fill=black!20]
\tikzstyle{matched}=[graphnode,fill=black!100]
\tikzstyle{matching} = [ultra thick]
\tikzset{
    >=stealth',
    pil/.style={
           ->,
           thick,
           shorten <=2pt,
           shorten >=2pt,}
}
\tikzset{->-/.style={decoration={
  markings,
  mark=at position .5 with {\arrow{>}}},postaction={decorate}}}

\makeindex

\newcommand{\argmax}{\text{arg max}}

\newtheorem{observation}[theorem]{Observation}
\newtheorem{assumption}[theorem]{Assumption}

\usepackage{thmtools,thm-restate} 

\newcommand{\E}{\mathbb{E}}

\newcommand{\OPT}{\textsc{OPT}\xspace}

\newcommand{\F}{{\mathcal{F}\xspace}}

\def\bx {{\bf x}}
\def\by {{\bf y}}
\def\bz {{\bf z}}

\def \reals {\mathbb{R}}



%% file: abstract.tex

Suppose there are $n$ Markov chains and we need to pay a per-step \emph{price} to advance them. The ``destination'' states of the Markov chains contain rewards; however, we can only get  rewards for a subset of them that satisfy a combinatorial constraint, e.g., at most $k$ of them, or  they are acyclic in an underlying graph. What strategy should we choose to advance the Markov chains if our goal is to maximize the total reward \emph{minus} the total  price that we pay?

In this paper we introduce a Markovian price of information model to capture  settings such as the above, where the input parameters of a combinatorial optimization problem are  given  via  Markov chains. We design optimal/approximation algorithms that jointly optimize the value of the combinatorial problem and the total paid price. We also study \emph{robustness} of our algorithms to the distribution parameters and how to handle the \emph{commitment} constraint.

Our work brings
  together two classical lines of investigation: getting optimal
  strategies for Markovian multi-armed bandits, and getting exact and approximation algorithms
  for discrete optimization problems using combinatorial as 
  well as linear-programming relaxation ideas.

\IGNORE{
In this paper we introduce  a Markovian price of information model where the input parameters of a combinatorial optimization problem are  given  via  Markov chains. 
To find the true value of a parameter, we need to perform \emph{multiple probes} on its Markov chain, where each probe incurs  a price and results in a random transition. 
We design optimal/approximation algorithms that jointly optimize the value of the combinatorial problem and the total  price paid.
We also study \emph{robustness} of our algorithms to the distribution parameters and how to handle a \emph{commitment} constraint.
}

%% file: introIPCO.tex

\section{Introduction}\label{sec:intro}

Suppose we are running an oil company
and are deciding where to set up new drilling operations.
There are several  candidate sites,
but the value of drilling each site is a random variable.
We must therefore \emph{inspect} sites before drilling.
Each inspection gives more information about a site's value,
but the inspection process is costly.
Based on laws, geography, or availability of equipment,
there are constraints on which sets of drilling sites are feasible.
We ask:
\begin{quote}
  What adaptive inspection strategy should we adopt to
   find a feasible set of sites to drill which maximizes,
   in expectation,
  the value of the chosen (drilled) sites minus the total inspection cost of all sites?
\end{quote}
Let us consider the optimization challenges in this problem:
\begin{enumerate}[(i)]
\item
  \label{item:combo}
  Even if we could fully inspect each site for free,
  choosing the best feasible set of sites
  is a \emph{combinatorial optimization} problem.
\item
  \label{item:multi}
  Each site may have \emph{multiple stages} of inspection.
  The costs and possible outcomes of later stages
  may depend on the outcomes of earlier stages.
  We use a \emph{Markov chain} for each site to model
  how our knowledge about the value of the site stochastically evolves with each inspection.
\item
  \label{item:robust}
  Since a site's Markov chain model may not exactly match reality,
  we want a \emph{robust} strategy
  that performs well even under small changes in the model parameters.
\item
  \label{item:commit}
  If there is competition among several companies,
  it may not be possible to do a few stages of inspection at a given site,
  abandon that site's inspection to inspect other sites,
  and then later return to further inspect the first site.
  In this case the problem has additional ``take it or leave it''
  or \emph{commitment} constraints,
  which prevent interleaving inspection of multiple sites.
\end{enumerate}

While each of the above aspects has been individually studied in
the past, no prior work addresses all of them.
In particular,  aspects \ref{item:combo} and \ref{item:multi}
have not been simultaneously studied before.
In this work we advance the state of the art by solving
the  \ref{item:combo}-\ref{item:multi}-\ref{item:robust} and
the \ref{item:combo}-\ref{item:multi}-\ref{item:commit} problems.

To study aspects \ref{item:combo} and \ref{item:multi} together,
in \S\ref{subsec:Model}
we propose the \emph{Markovian Price of Information} (\MPOI) model.
The \MPOI model unifies prior models
which address  \ref{item:combo} or \ref{item:multi} alone.
These prior models include those of
Kleinberg et al.~\cite{KWW-EC16} and Singla~\cite{Singla-SODA18},
who study the combinatorial optimization aspect~\ref{item:combo}
in the so-called \emph{price of information} model,
in which each site has just a single stage of inspection;
and those of Dimitriu et al.~\cite{DTW-SIDMA03} and Kleinberg et al.~\cite[Appendix~G]{KWW-EC16},
who consider the multiple stage inspection aspect~\ref{item:multi}
for the problem of selecting just a single site.

Our main results show how to solve combinatorial optimization problems,
including both maximization and minimization problems,
in the \MPOI model.
We give two methods of transforming classic algorithms,
originally designed for the \FreeInfo (inspection is free) setting,
into \emph{adaptive} algorithms for the \MPOI setting.
These adaptive algorithms respond dynamically to the random outcomes of inspection.
\begin{itemize}[noitemsep,topsep=5pt]
\item
  In \S\ref{sec:UtilityMaximization}
  we  transform ``greedy'' $\alpha$-approximation
  algorithms in the \FreeInfo setting
  into $\alpha$-approximation adaptive algorithms in the \MPOI setting
  (Theorem~\ref{thm:frugalToAdaptiveMax}).
  For example, this yields optimal algorithms for matroid optimization
  (Corollary~\ref{cor:impCorollaries}).
\item
  In \S\ref{subsubsec:RobustnessModel}
  we show how to slightly modify
  our $\alpha$-approximations for the \MPOI setting in Theorem~\ref{thm:frugalToAdaptiveMax}
  to make them robust to small changes in the model parameters
  (Theorem~\ref{thm:RobustInformal}).
\item
  In \S\ref{subsec:IntroCommitment}
  we   use   \emph{online contention resolution schemes} (OCRSs)~\cite{FSZ-SODA16} to   transform LP based \FreeInfo maximization algorithms
into adaptive \MPOI algorithms while respecting the commitment constraints.
Specifically, a $1/\alpha$-selectable OCRS yields $\alpha$-approximation
with commitment  (Theorem~\ref{thm:OCRStoCommitment}).

\IGNORE{  in the \FreeInfo setting,\footnote{%
    By this we mean a combinatorial maximization problem in which
    each item's value is a random variable which can be revealed for free,
    but we must irrevocably take or leave an item
    immediately after revealing its value.
    In fact, OCRSs treat a harder version of the problem in which
    an adversary can choose the order in which values are revealed.}
    }
\end{itemize}

The general idea behind our first result (Theorem~\ref{thm:frugalToAdaptiveMax}) is the following.
A \frugal combinatorial algorithm (Definition~\ref{defn:frugalPacking})
is, roughly speaking, ``greedy'':
it repeatedly selects the feasible item of greatest marginal value.
We show how to adapt \emph{any} \frugal algorithm to the \MPOI setting:
\begin{itemize}[noitemsep,topsep=5pt]
\item
  Instead of using a fixed value for each item~$i$,
  we use a \emph{time-varying ``proxy'' value}
  that depends on the state of $i$'s Markov chain.
\item
  Instead of immediately selecting the item~$i$ of greatest marginal value,
  we \emph{advance $i$'s Markov chain one step}.
\end{itemize}
The main difficulty lies in choosing each item's proxy value,
for which simple heuristics can be suboptimal.
We use a quantity for each state of each item's Markov chain
called its \emph{grade},
and an item's proxy value is its \emph{minimum grade so far}.
A state's grade is closely related to the Gittins index
from the multi-armed bandit literature,
which we discuss along with other related work in~\S\ref{sec:relatedWork}.

\section{The Markovian Price of Information Model} \label{subsec:Model}

To capture the evolution of our knowledge about an item's value,
we use the notion of a Markov system from~\cite{DTW-SIDMA03}
(who did not consider values at the destinations).
\begin{definition} [Markov System]
  A Markov system $\MS=(V,P,s,T, \bprice, \breward)$ for an element
  consists of a discrete Markov chain with state space $V$, a transition
  matrix $P=\{p_{u,v}\}$ indexed by $V \times V$ (here $p_{u,v}$ is the
  probability of transitioning from $u$ to $v$), a starting state $s$, a
  set of absorbing \emph{destination} states
  $T \subseteq V$, a non-negative probing price
  $\price^u\in \reals_{\geq 0}$ for every state $u \in V \backslash T$,
  and a value $\reward^t \in \reals$ for each destination state
  $t \in T$. We assume that every state $u \in V$ reaches some
  destination state.
\end{definition}

We have a collection $J$ of \emph{ground elements}, each associated with
its own Markov system.  An element is \emph{\prepared} if its Markov
system has reached one of its absorbing destination states. For a \prepared
element, if $\traj$ is the (random) \emph{trajectory} of its Markov chain then $\dest(\traj)$
denotes its associated destination state.
We now define the \MPOI game, which consists of an objective function on
  $J$.
\begin{definition} [\MPOI Game]
  \label{defn:MPOIGame}
  Given a set of ground elements $J$, constraints $\F \subseteq 2^J$,
  an objective function $f:2^J \times \R^{|J|} \rightarrow \R$, and a
  Markov system $\MS_i=(V_i,P_i,s_i,T_i,\bprice_i,\breward_i)$ for each
  element $i \in J$, the \MPOI game is the following.  At each time
  step, we either advance a Markov system $\MS_i$ from its current state
  $u \in V_i\setminus T_i$ by incurring price $\price^u_i$, or we end the
  game by selecting a subset of \emph{\prepared} elements
  $\I \subseteq J$ that are \emph{feasible}---i.e.,  $\I \in \F$.
\end{definition}
A common choice for $f$ is the \emph{additive} objective
$f(\I,{\bf x}) = \sum_{i\in \I} x_i$.

Let $\trajProf$ denote the \emph{trajectory profile} for the \MPOI game: it
consists of the random trajectories~$\traj_i$ taken by all the Markov chains~$i$ at
the end of the game.  To avoid confusion, we write the selected feasible
solution $\I$ as $\I(\trajProf)$.
A utility/disutility optimization problem is to give a strategy
for a \MPOI game while optimizing both the objective and the total
price.

\medskip
\noindent \textbf{Utility Maximization  (\UtiMax)}: A \MPOI game where the
constraints $\F$ are \emph{downward-closed} (i.e., \emph{packing}) and
the values $\breward_i$ are non-negative for every $i \in J$ (i.e.,
$\forall t\in T_i$, $\reward^t_i \geq 0$, and can be understood as
a reward obtained for selecting   $i$).  The goal is to
find a strategy $\ALG$ maximizing  \emph{utility}:
\begin{align} \label{eq:MPOI}
\Umax(\ALG) \overset{\Delta}{=}  \E_{\trajProf} \Big[ \underbrace{f
  \left(\I(\trajProf), \{\reward^{\dest(\traj_i)}_i\}_{i \in \I(\trajProf)}
  \right)}_{\text{value}}  - \underbrace{ \textstyle \sum_i  \sum_{u\in
  \traj_i} \price_i^u}_{\text{total price}}    \Big].
\end{align}
Since the empty set is always
feasible, the optimum utility is non-negative.
\medskip

We also define a minimization variant of the problem that is useful to
capture covering combinatorial problems such as minimum spanning trees and set cover.

\medskip
\noindent \textbf{Disutility Minimization  (\DisMin) }: A \MPOI game where
the constraints $\F$ are \emph{upward-closed} (i.e., \emph{covering}) and the
values $\breward_i$ are non-negative for every $i \in J$ (i.e.,
$\forall t\in T_i$, $\reward^t_i \geq 0$, and can be understood as a cost
we pay for selecting  $i$).  The goal is to
find a strategy $\ALG$ minimizing \emph{disutility}:
\[\Umin(\ALG) \overset{\Delta}{=}  \E_{\trajProf} \Big[ f \left(\I(\trajProf),
    \{\reward^{\dest(\traj_i)}_i\}_{i \in \I(\trajProf)} \right)  + {\textstyle \sum_i
  \sum_{u\in \traj_i} \price_i^u }    \Big].
\]
We will assume that the function $f$ is non-negative when all
$\breward_i$ are non-negative. Hence, the disutility of the optimal
policy is non-negative.
\medskip

In the special case where all the Markov chains for a \MPOI game are formed by a \emph{directed acyclic graph} (\DAG), we call the corresponding  optimization problem \DAG-\UtiMax or \DAG-\DisMin.

\section{Adaptive Utility Maximization via \frugal Algorithms}
\frugal algorithms, introduced in Singla~\cite{Singla-SODA18},
 capture the intuitive notion of ``greedy'' algorithms.
There are many known  \frugal algorithms, e.g., optimal algorithms for matroids and $O(1)$-approx algorithms for matchings, vertex cover, and facility location.
These \frugal algorithms were designed
in the traditional \emph{free information} (\FreeInfo) setting, where
 each ground element has a  fixed value. Can we use them in the \MPOI world?

Our main contribution is a technique that
adapts \emph{any} \frugal algorithm to the \MPOI world,
achieving the \emph{same approximation ratio} as the original algorithm.
The result applies to \emph{semiadditive} objective functions~$f$,
which are those of the form $f(\I,{\bf x}) = \sum_{i\in \I} x_i + h(\I)$
for some $h:2^J \rightarrow \R$.

\begin{restatable}{theorem}{frugalToAdaptiveMax} \label{thm:frugalToAdaptiveMax}
For a semiadditive objective function $\val$, if there exists an  $\alpha$-approximation \frugal algorithm for a \UtiMax problem over some packing constraints $\F$ in the \FreeInfo world, then there exists an $\alpha$-approximation strategy for the corresponding \UtiMax problem in the \MPOI world.
\end{restatable}

We prove an analogous result for \DisMin in \S\ref{sec:DisutilityMinimization}.
The following corollaries immediately follow from known \frugal algorithms~\cite{Singla-SODA18}.

\begin{corollary} \label{cor:impCorollaries}
In the \MPOI world, we have:
  \begin{itemize} [noitemsep,topsep=5pt]
\item An optimal algorithm for both \UtiMax and \DisMin for matroids.
\item A $2$-approx for \UtiMax for matchings and a $k$-approx for a $k$-system.
\item A $\min\{f,\log n\}$-approx for \DisMin for set-cover, where $f$ is the maximum number of sets in which a ground element is present.
\item  A $1.861$-approx for \DisMin for  facility location.
\item  A $3$-approx for \DisMin for prize-collecting Steiner tree.
  \end{itemize}
\end{corollary}

\input{preliminaryIPCO}

\input{utilitymax}
\section{Robustness in Model Parameters}
\label{subsubsec:RobustnessModel}

In practical applications,
the parameters of Markov systems (i.e., transition probabilities, values, and prices)
are not known exactly but are \emph{estimated} by statistical sampling.
In this setting,
the \emph{true parameters}, which govern how each Markov system evolves,
differ from the estimated parameters that the algorithm uses to make decisions.
This raises a natural question:
how well does an adapted \frugal algorithm do when
the true and the estimated parameters differ?
We would hope to design a \emph{robust} algorithm,
meaning small  estimation errors
cause only small  error in the utility objective.

In the important special case where  the Markov chain corresponding to each element is formed by a \emph{directed acyclic graph} (\DAG),
an adaptation of our strategy in Theorem~\ref{thm:frugalToAdaptiveMax} is robust.  This \DAG assumption turns out to be necessary as  similar results do not hold for general Markov chains (see Appendix~\ref{sec:DAGNecessaryforRobust}).
In particular, we prove the following generalization of Theorem~\ref{thm:frugalToAdaptiveMax} under the \DAG assumption.

\begin{theorem}[\textnormal{Informal statement}]
\label{thm:RobustInformal}
If  there exists an  $\alpha$-approximation \frugal algorithm $\A$ ($\alpha \geq 1$) for a packing  problem with a semiadditive objective function, then it suffices to estimate the true model parameters of a \DAG-\MPOI game within an additive error of $\epsilon/\poly$, where $\poly$ is some polynomial in the size of the input, to design a strategy with utility at least  $\frac{1}{\alpha} \cdot \OPT - \epsilon$, where \OPT is the utility  of the optimal policy that  knows all the \emph{true}  model parameters.
\end{theorem}

Specifically, our strategy $\ALGHA$ for Theorem~\ref{thm:RobustInformal} is obtained from the strategy in Theorem~\ref{thm:frugalToAdaptiveMax} by making use of the following  idea: each time we advance an element's Markov system, we slightly increase the estimated grade of every state in that Markov system.
This ensures that whenever we advance a Markov system, we  advance through an entire epoch and remain optimal in the ``teasing game".

Our analysis of $\ALGHA$ works roughtly as follows.
We first show that close estimates of the model parameters of a Markov system can be used to  closely estimate the grade of each state.
We can therefore assume that close estimates of all grades are given as input.
Next we define the ``shifted'' prevailing cost corresponding to the ``shifted'' grades. This allows us to equate the utility of $\ALGHA$ by the utility of running $\A$ in the ``modified'' surrogate problem where the input to $\A$ is the ``shifted'' prevailing costs instead of the \emph{true} prevailing costs.
Finally, we prove that the ``shifted'' prevailing costs are close to the real prevailing costs and thus the ``modified'' surrogate problem is close to the surrogate problem.
This allows us to bound the utility of running $\A$ in the ``modified'' surrogate problem by the optimal strategy to the surrogate problem.
Combining with Lemma~\ref{lem:boundAdapMax} finishes the proof of Theorem~\ref{thm:RobustInformal}.

Similar arguments extend to prove the analogous result for \DisMin.

\input{robustness}
\section{Handling Commitment Constraints}
\label{subsec:IntroCommitment}
Consider the \MPOI model defined in \S\ref{subsec:Model} with an additional restriction that whenever we abandon advancing a Markov system, we need to \emph{immediately} and \emph{irrevocably} decide if we are selecting this element into the final solution $\I$. Since we only select \prepared elements, any element that is not \prepared when we abandon its Markov system is automatically discarded.
We call this constraint \emph{commitment}.
The benchmark for our algorithm  is the optimal policy \emph{without} the commitment constraint.
For single-stage  probing,  such commitment constraints have been well studied, especially in the context of stochastic matchings~\cite{CIKMR09,BGLMNR-Algorithmica12}.

We study  \UtiMax    in the \DAG model with the commitment constraint.
Our algorithms make use of the \emph{online contention resolution schemes}~(OCRSs) proposed in~\cite{FSZ-SODA16}. OCRSs address  our problem in the \FreeInfo world\footnote{In fact, OCRSs consider a variant where  the adversary chooses the order in which the elements are tried. This handles the present problem   where we may choose the order.} (i.e., we can see the realization of the r.v.s for free, but there is the commitment constraint). Constant factor ``selectable" OCRSs are known for several  constraint families: $\frac14$~for matroids, $\frac{1}{2e}$ for matchings, and $\Omega(\frac{1}{k})$ for intersection of $k$ matroids~\cite{FSZ-SODA16}.
We  show how to adapt them
to  \MPOI  with commitment. 

\begin{restatable}{theorem}{OCRStoCommitment}
\label{thm:OCRStoCommitment}
  For an additive objective, if there exists a $1/\alpha$-selectable OCRS ($\alpha \geq 1$)  for a packing constraint $\F$,
  then there exists an $\alpha$-approximation algorithm
  for the corresponding \DAG-\UtiMax problem with  commitment.
\end{restatable}

The proof of this result uses a new LP relaxation (inspired from~\cite{GM-STOC07})  to  bound the optimum utility  of a \MPOI game \emph{without} commitment (see \S\ref{sec:LPBoundOPT}). Although this relaxation is not exact even for Pandora's box (and cannot be used to design optimal strategies in Corollary~\ref{cor:impCorollaries}), it turns out to suffice for our approximation guarantees. In \S\ref{sec:roundingLP}, we use an OCRS to round this LP with only a small loss in the utility, while respecting the commitment constraint. 

\begin{remark}
We do not consider \DisMin problem under commitment  because it captures prophet inequalities in a minimization setting where no polynomial approximation is possible even for i.i.d. r.v.s~\cite[Theorem~$4$]{EHLM-JDM17}.
\end{remark}

\input{promptness}
\section{Related Work}	\label{sec:relatedWork}

Our work is related to work on multi-armed bandits in the scheduling literature.
The Gittins index theorem~\cite{GJ-Journal74} provides a simple optimal strategy for several scheduling problems where the objective is to maximize the long-term exponentially discounted
reward. This theorem turned out to be  fundamental and
  \cite{Tsitsiklis-Prob94,Weber-Prob92,Whittle-Stat80} gave alternate proofs.  It can be also used to solve Weitzman's  Pandora's box.
The reader is referred to the  book~\cite{GGW-Book11} for further discussions on this topic.
Influenced by this literature,
\cite{DTW-SIDMA03} studied  scheduling of Markovian jobs, which is a  minimization variant of the Gittins index theorem without any discounting.
Their paper is part of the inspiration for our \MPOI model.

\eat{The field of combinatorial optimization has been extensively studied: we refer the readers to Schrijver's popular book~\cite{Schrijver-Book03}, and the references therein. In recent years,
there has also been a lot of interest in  studying these combinatorial problems  for stochastic inputs.
\cite{DGV-FOCS04,DGV05,GM-SODA07TALG12,GM-STOC07,BGK-SODA11,LiYuan-STOC13,Ma-SODA14}  considered  stochastic knapsack, \cite{CIKMR09,A11,BGLMNR-Algorithmica12,BCNSX15,AGM15} studied stochastic matchings, \cite{GuhaM09,GKNR-SODA12,BN-IPCO14} studied stochastic orienteering,  \cite{ANS-WINE08,GN-IPCO13,ASW14,GNS-SODA17,GNS-SODA16} considered stochastic submodular maximization, and \cite{GM-STOC07,GuhaM09,GKMR-FOCS11,Ma-SODA14} studied budgeted multi-armed bandits.  These works (besides~\cite{GM-STOC07}) do not consider a mixed-sign utility objective or a multi-stage probing, which is the focus of this paper.}
%

The Lagrangian variant of stochastic probing considered in~\cite{GM-STOC07} is similar to our \MPOI model. However, their approach  using an LP relaxation to design a probing strategy is fundamentally different from our approach using a \frugal algorithm. E.g., unlike  Corollary~\ref{cor:impCorollaries}, their approach cannot give \emph{optimal} probing strategies for matroid constraints due to an integrality gap. Also, their approach does not work for \DisMin.
In \S\ref{subsec:IntroCommitment}, we extend their techniques using OCRSs to handle the commitment constraint for \UtiMax.

 There is also a large body of work in related  models where information has a price~\cite{GK-FOCS01,CFGKRS-Journal02,KK-SODA03,GMS-Transactions07,CJK+-AAAI15,AbbasH-Book15,CHHKK-COLT15,CILSZ-ArXiv17}.
Finally, as discussed in the introduction, the works in~\cite{KWW-EC16} and \cite{Singla-SODA18} are directly relevant to this paper. The former's primary  focus is  on \emph{single item} settings and its applications to auction design, and the latter studies price of information in a \emph{single-stage}  probing model.
Our contributions concern selecting \emph{multiple items} in \emph{multi-stage} probing model,
in some sense unifying these two lines of work.

\input{relatedwork}


\IGNORE{
\subsection{Organization of the Paper}
In \S\ref{sec:Preliminaries}, we explain the  concepts of grade and prevailing cost that form the key to our arguments.
Then in \S\ref{sec:UtilityMaximization}, we formally define  a \frugal algorithm and
show how to use it to obtain a good adaptive strategy for the \UtiMax problem. The corresponding proofs for  \DisMin  are in the appendix.
\eat{The \DisMin problem under the \MPOI model is handled in \S\ref{sec:DisutilityMinimization}.}
\eat{\S\ref{sec:Robustness}  discusses how to make our algorithms robust to input parameters.
Finally, in \S\ref{sec:CommitmentConstraints}, we
handle the commitment constraint using a new LP and OCRSs.}
}


%% file: preliminaryIPCO.tex



Before proving Theorem~\ref{thm:frugalToAdaptiveMax},
we define a \emph{grade} for every state in a Markov system in
\S\ref{para:grade}, much as in~\cite{DTW-SIDMA03}. This grade is a variant of the popular  \emph{Gittins index}. In \S\ref{para:PrevailingCost}, we use the grade  to define a \emph{prevailing cost} and an \emph{epoch} for a trajectory. In \S\ref{sec:UtilityMaximization}, we use these definitions to prove Theorem~\ref{thm:frugalToAdaptiveMax}.
We consider \UtiMax throughout,
but analogous definitions and arguments hold for \DisMin.

\subsection{Grade of a State} \label{para:grade}
To define the \emph{grade} $\tau^v$ of a state $v \in V$
in Markov system $\MS=(V,P,s,T,\bprice,\breward)$,
we consider the following Markov game called \emph{$\tau$-penalized $\MS$},
denoted $\MS(\tau)$.
Roughly, $\MS(\tau)$ is the same as $\MS$
but with a \emph{termination penalty}, which is a constant $\tau \in \reals$.

Suppose $v\in V$ denotes the current state of $\MS$ in the game $\MS(\tau)$. In each move,  the player has two choices: (a)   \emph{Halt} that immediately ends the game, and (b)   \emph{Play} that changes the  state, price, and value as follows:
  \begin{itemize} [noitemsep,topsep=5pt]
\item     If $v \in V \setminus T$,  the player pays price $\price^v$,
    the current  state  of $\MS$  changes according to the transition matrix $P$,
    and the game continues.
\item     If $v \in T$, then the player receives \emph{penalized value} $r^v - \tau$,
    where $\tau$ is the aforementioned termination penalty,
    and the game ends.
  \end{itemize}

The player wishes to maximize his \emph{utility},
which is the expected value he obtains minus the expected price he pays.
We write $U^v(\tau)$ for the utility attained by optimal play
starting from state $v \in V$.

The utility $U^v(\tau)$ is clearly non-increasing in the penalty~$\tau$,
and one can also show that it is continuous~\cite[Section~4]{DTW-SIDMA03}.
In the case of  large penalty $\tau \to +\infty$,
it is optimal to halt immediately, achieving $U^v(\tau) = 0$.
In the opposite extreme $\tau \to -\infty$,
it is optimal to play until completion,
achieving $U^v(\tau) \to +\infty$.
Thus, as we increase $\tau$ from $-\infty$ to $+\infty$,
the utility $U^v(\tau)$  becomes~$0$
at some critical value $\tau = \tau^v$. This critical value $\tau^v$ that  
depends on  state~$v$ is the \emph{grade}.

\begin{definition}[Grade]
  The \emph{grade} of a state $v$ in Markov system $\MS$ is
$
    \tau^v \overset{\Delta}{=}  \sup \{\tau \in \R \mid U^v(\tau) > 0\}.
$
For a \UtiMax problem, we write the grade of a state $v$ in  Markov system $\MS_i$ corresponding to element~$i$  as $\tau_i^v$.
\end{definition}

  The quantity grade of a state is well-defined  from the above discussion.
We emphasize that it  is independent of all other Markov systems.
Put another way, the grade of a state is
the penalty  $\tau$ that makes the player
\emph{indifferent} between halting and playing.
It is known how to compute grade
efficiently~\cite[Section~7]{DTW-SIDMA03}.

\subsection{Prevailing Cost and Epoch}
\label{para:PrevailingCost}
We now define a  \emph{prevailing cost}~\cite{DTW-SIDMA03} and an \emph{epoch}.
The prevailing cost of Markov system $\MS$ is its minimum grade at any point in time.

\begin{definition}[Prevailing Cost]
\label{defn:PrevailingCostMax}
The \emph{prevailing cost} of Markov system $\MS_i$ in a trajectory $\traj_i$ is
$\Ymax(\traj_i) = \min_{v \in \traj_i}\{\tau_i^{v} \}$.
For trajectory profile $\trajProf$, denote  $\Ymax(\trajProf)$   the list of prevailing costs for each Markov system.
\end{definition}

Put another way,
the prevailing cost is
the maximum termination penalty  for the game $\MS(\tau)$
such that for every state along $\traj$  the player does not want to halt.


Observe that the prevailing cost of a trajectory can only decrease as it extends further. In particular, it decreases whenever the Markov system reaches a state with grade smaller than each of the previously visited states. We can therefore view the prevailing cost as a non-increasing piecewise  constant function of time.
This motivates us to define an epoch. 

\begin{definition}[Epoch]
\label{defn:EpochMax}
An \emph{epoch} for a trajectory $\traj$ is  any maximal continuous segment  of $\traj$ where  the prevailing cost does not change.
\end{definition}
%
Since the grade can be computed efficiently, we can also compute the prevailing cost and epochs of a trajectory   efficiently.


%% file: utilitymax.tex

\subsection{Adaptive Algorithms for Utility Maximization}
\label{sec:UtilityMaximization}

In this section, we prove  Theorem~\ref{thm:frugalToAdaptiveMax} that  adapts a  \frugal algorithm in
 \FreeInfo world to a probing strategy in the \MPOI world.
 This theorem concerns \emph{semiadditive functions},
which are useful to capture non-additive objectives of problems like facility location and prize-collecting Steiner tree.

\begin{definition}[Semiadditive Function~\cite{Singla-SODA18}]
A function $f(\I,\X): 2^{J}\times \R^{|J|} \rightarrow \R$ is \emph{semiadditive} if there exists a function $h:2^J \rightarrow \R$  s.t.
$	f(\I,{\bf x}) = \sum_{i\in \I} x_i + h(\I).
$
\end{definition}

All additive functions are semiadditive with $h(\I)=0$ for all $\I$. To capture the facility location problem on a graph $G=(J,E)$ with metric $(J,d)$,  clients $C\subseteq J$, and facility opening costs ${\bf x}: J \rightarrow \reals_{\geq 0}$,  we can define $h(\I) = \sum_{j \in C} \min_{i \in \I} d(j,i)$. Notice $h$ only depends on the identity of facilities $\I$ and not their opening costs.
%


The proof of Theorem~\ref{thm:frugalToAdaptiveMax} takes two  steps. We first give a randomized reduction to upper bound the utility of the optimal strategy in the \MPOI world with the optimum  of a \emph{surrogate problem} in the \FreeInfo world. 
Then,  we  transform a \frugal algorithm into a strategy with utility close to this  bound.

\subsubsection{Upper Bounding the Optimal Strategy Using a Surrogate.}
\label{subsec:BoundOptMax}

The \emph{main idea} in this section is to show that for \UtiMax,  no strategy (in particular,  optimal) can derive more utility from an element $i \in J$ than its  prevailing cost.  Here, the prevailing cost of $i$ is  for a random trajectory to a destination state in Markov system $\MS_i$.
Since the optimal strategy can only select a feasible set in $\F$, this idea naturally leads  to the following \FreeInfo \emph{surrogate  problem}: imagine each element's value is exactly its (random) prevailing cost, the goal is to select a  set feasible in $\F$ to maximize the  total value. In Lemma~\ref{lem:boundAdapMax}, we show that the expected optimum value of this surrogate problem is an upper bound on the optimum utility for \UtiMax.
First, we formally define the surrogate problem.



\begin{definition}[Surrogate Problem] \label{defn:SurMax}
Given a \UtiMax problem with semiadditive objective  $\val$ and packing constraints $\F$ over  universe $J$, the corresponding  \emph{surrogate} problem over $J$ is the following. It consists of constraints  $\F$ and (random) objective function  $\tilde{f}:2^J \rightarrow \reals $  given by $\tilde{f}(\I) = \val(\I, \bYmax(\trajProf))$, where $\bYmax(\trajProf)$ denotes the prevailing costs over  a random trajectory profile $\trajProf$ consisting of independent random trajectories for each element $i\in J$  to a destination state. The goal is to select  $\I \in \F$ to maximize $\tilde{f}(\I) $.
\end{definition}

Let $\Sur(\trajProf)\overset{\Delta}{=} \max_{\I \in \F} \{ \val(\I, \bYmax(\trajProf))\}$ denote the optimum value of the surrogate problem for trajectory profile $\trajProf$.
We now upper bound the optimum utility  in the \MPOI world.  Our proof borrows ideas from the ``prevailing reward argument'' in~\cite{DTW-SIDMA03}.

\begin{restatable}{lemma}{boundAdapMax} \label{lem:boundAdapMax}
For a \UtiMax problem with objective $\val$ and packing constraints $\F$, let  $\OPT$ denote the utility of the optimal strategy. Then,
\[ \textstyle {\OPT \quad \leq \quad \E_{\trajProf} [ \Sur(\trajProf)] \quad = \quad \E_{\trajProf} \big[\max_{\I \in \F} \{ \val(\I, \bYmax(\trajProf))\}  \big], }
\]
where the expectation is  over a random trajectory profile $\trajProf$ that has every Markov system reaching a  destination state.
\end{restatable}

We prove Lemma~\ref{lem:boundAdapMax} in \S\ref{misproof:boundAdapMax}.

\subsubsection{Designing an Adaptive Strategy  Using a Frugal Algorithm.}
\label{subsec:FrugalAlgMax}

A \frugal algorithm selects elements one-by-one and irrevocably. Besides greedy algorithms, its definition also captures ``non-greedy'' algorithms such as  primal-dual algorithms that do not have the reverse-deletion step~\cite{Singla-SODA18}.

\begin{definition}[\frugal Packing  Algorithm]\label{defn:frugalPacking}
For a combinatorial optimization problem on universe $J$ in the \FreeInfo  world with packing constraints $\F \subseteq 2^J$ and objective $f: 2^J \rightarrow \reals$, we say Algorithm \A is \frugal if  there exists a \emph{marginal-value}  function $g(\Y,i,y):\reals^J \times J \times \reals \rightarrow \reals$ that is increasing in $y$, and for which the pseudocode is given by Algorithm~\ref{alg:frugalPacking}. Note that this algorithm always returns a feasible solution if  $\emptyset \in \F$.
\setlength{\intextsep}{5pt}
\begin{algorithm}
\caption{\frugal Packing Algorithm \A}
\label{alg:frugalPacking}
\begin{algorithmic}[1]
\State Start with $M=\emptyset$ and  $v_i=0$ for each element $i \in J$.
\State For each element $i\not\in M$, compute  $v_i = g( \Ym, i,Y_i)$. Let $j = \argmax_{i\not \in M ~\&~ M\cup i \in \F} \{v_i\}$. \label{alg:FrugalcomputeVi}
\State If $v_j>0$ then add $j$ into $M$ and  go to Step~\ref{alg:FrugalcomputeVi}. Otherwise, return $M$.
\end{algorithmic}
\end{algorithm}
\end{definition}
\vspace{-0.5cm}

The following lemma shows that a \frugal algorithm can be converted to a strategy with the same utility in the \MPOI world. 

\begin{restatable}{lemma}{convertFrugalAlgMax} \label{lem:convertFrugalAlgMax}
Given a  \frugal packing Algorithm~\A, there exists an adaptive strategy $\ALG_\A$ for the corresponding \UtiMax problem  in  \MPOI world with utility at least
$ \E_{\trajProf} [\val(\A(\bYmax(\trajProf)), \bYmax(\trajProf))],
$
where $\A(\bYmax(\trajProf)$ is the solution returned by  \A for objective $f(\I) = \val(\bYmax(\trajProf),\I)$.
\end{restatable}

We prove Lemma~\ref{lem:convertFrugalAlgMax} in \S\ref{misproof:convertFrugalAlgMax}.
Finally, we can prove Theorem~\ref{thm:frugalToAdaptiveMax}.

\begin{proof}[Proof of Theorem~\ref{thm:frugalToAdaptiveMax}]
From Lemma~\ref{lem:convertFrugalAlgMax},  the utility of $\ALG_\A$ is at least
$ \E_{\trajProf} [\val(\A(\bYmax(\trajProf)), \bYmax(\trajProf))]. $
Since Algorithm~\A is an $\alpha$-approx algorithm in the \FreeInfo world, it follows
\begin{align*}
\E_{\trajProf} [\val(\A(\bYmax(\trajProf)), \bYmax(\trajProf))] \geq \frac{1}{\alpha} \cdot \E_{\trajProf} \big[\max_{\I \in \F} \{ \val(\I, \bYmax(\trajProf))\} \big].
\end{align*}
Using the upper bound on  optimal utility $\OPT \leq \E_{\trajProf} \big[\max_{\I \in \F} \{ \val(\I, \bYmax(\trajProf))\} \big]$ from Lemma~\ref{lem:boundAdapMax}, we  have  utility of $\ALG_\A$ is at least $\frac{1}{\alpha} \cdot \OPT$.
\end{proof}

In \S\ref{sec:DisutilityMinimization}, a similar approach is used for the \DisMin problem with semi-additive function.
This shows that for both \UtiMax or \DisMin problem with semi-additive function, a $\frugal$  algorithm can be transformed from \FreeInfo to \MPOI world while retaining its performance.

%% file: robustness.tex

We  formally state our main theorem and the parameters on which it depends in Section~\ref{subsec:ParametersAssumptionsDAG}.
Section~\ref{subsubsec:ClosenessImpliesGradeMax} shows that close estimates of transition probabilities can be used to obtain close estimates of the grades.
In Section~\ref{subsubsec:FrugalRobusttoMax}, we use these estimated grades  to transform a \frugal algorithm  into a robust adaptive algorithm for  \DAG-\UtiMax.
Similar arguments can be used to  obtain the corresponding results for \DAG-\DisMin  (we omit this proof).

\subsection{Main Results and Assumptions}
\label{subsec:ParametersAssumptionsDAG}

We first explicitly define the \emph{input size} of \DAG-\UtiMax as follows.
\begin{enumerate}[label=(\roman*)]
\item $n$ is the number of Markov systems.
\item $k$ is the maximum number of elements in a feasible solution, i.e., $k\overset{\Delta}{=} \max_{\I \in \F} |\I|$.
\item $D$ is the maximum depth of any \DAG Markov system.
\end{enumerate}

Denote $B$ an upper bound on all input prices and values, i.e., $\forall i, \forall \pi \in \bprice_i, \forall r \in \breward_i$, we have $|\pi| \leq B, |r|\leq B$. We make the following assumption.
\begin{assumption}\label{assump:BoundedInput}
The upper bound $B$ is polynomial in $n,k$, and $D$.
\end{assumption}
Such an assumption turns out to be necessary (see Appendix~\ref{sec:BoundedInput}).
We now state our main theorem of this section.

\begin{restatable}{theorem}{FrugalRobusttoAdaptiveMax} \label{thm:FrugalRobusttoAdaptiveMax}
Consider a \DAG-\UtiMax problem with a semiadditive objective  and satisfying Assumption~\ref{assump:BoundedInput}. Suppose there exists an $\alpha$-approximation \frugal algorithm in the \FreeInfo world.
If each input parameter is known to within an additive error of $\epsilon/\poly$, where $\poly$ is some polynomial in  $n,k$, and $D$, then there exists an adaptive algorithm \ALGH with utility at least
$$
\frac{1}{\alpha} \cdot \OPT - \epsilon,
$$
where \OPT is the utility of the optimal policy that exactly knows  the true input parameters.
\end{restatable}


To simplify the proof of Theorem~\ref{thm:FrugalRobusttoAdaptiveMax}, we also assume the following  without loss of generality (see  Appendix~\ref{sec:JustifyAssumptionsRobust} for justifications).
\begin{enumerate}[label=(\roman*)] \setcounter{enumi}{3}
\item All non-zero transition probabilities are lower bounded by $1/P$, where $P$ is a polynomial in $n,k$, and $D$.
\item We know the  prices $\bprice$ and the rewards $\breward$ exactly, i.e., the only unknown input parameters are the transition probabilities. \label{assum:rewardsExact}
\end{enumerate}


\subsection{Well-Estimated Input Parameters Imply Well-Estimated Grades}
\label{subsubsec:ClosenessImpliesGradeMax}
We  call the set of Markov systems constructed using our estimated transition probabilities the \emph{estimated world}.
The $i$th Markov system in this \emph{estimated world} is denoted by $\MCH_i=(V_i,\PH_i,s_i,T_i,\bprice_i,\breward_i)$, where $\PH_i$ contains the estimated transition probabilities.  Note,   $\bprice_i$ and $\breward_i$ are exact due to Assumption~\ref{assum:rewardsExact}.
We estimate the grade of a state by simply computing the grade of that state in the estimated world.
The following Lemma~\ref{lem:WellEstGradeMax} bounds the error in estimated grades in terms of the error in  transition probabilities.

\begin{lemma} \label{lem:WellEstGradeMax}
Consider the \DAG-\UtiMax problem satisfying the assumptions in Section~\ref{subsec:ParametersAssumptionsDAG}.
Suppose all  transition probabilities are estimated to within an additive error of $\epsilon < 1/P$,
then $\forall i, \forall u \in V_i$, the estimated grade $\tauH^u_i$ is within an additive factor of $ O(L \cdot \epsilon)$ from the real grade $\tau^u_i$, where $L= D^2 BP$.
\end{lemma}
\begin{proof}
\eat{Consider an arbitrary state $u$ in the $i$th Markov system.
Denote $\tauH^u_i$ the estimated grade of $u$ (i.e., the grade of $u$ in the estimated world).}
We  show below that $\tau^u_i \geq \tauH^u_i - L \cdot \epsilon$.
A symmetrical argument shows $\tauH^u_i \geq \tau^u_i - L \cdot \epsilon$, which finishes the proof of this lemma.


We consider the Markov game $\GH_u$ defined in Section~\ref{para:grade} in the estimated world.
By definition, there exists an optimal policy \pol that advances the chain at least one more step and achieves an expected utility of 0.
Also consider the Markov game $G_u$ in the real world and apply \pol in $G_u$.
Notice \pol might be sub-optimal in $G_u$ and might therefore obtain a negative expected value.
Let $\tau_{fair}$ be the cost $\tau$ in $G_u$ such that \pol obtains an expected value of 0.
It follows that $\tau^u_i \geq \tau_{fair}$.
It therefore suffices to show that $\tau_{fair} \geq \tauH^u_i - L \cdot \epsilon$.

Denote the set of trajectories when applying \pol (in either world) by $\SC$ and those in which the item is picked by $\SC_{win}$.
Denote $p_\trajProf$ the probability of a trajectory $\trajProf \in \SC$ in the real world
and $\pH_\trajProf$ the probability of it in the estimated world.
Let $r_\trajProf$ be the utility of $\trajProf$ (as defined for \UtiMax by ignoring the cost $\tau$) in either world.
It follows that
$$
\tau_{fair} = \frac{1}{\sum_{\trajProf \in \SC_{win}} p_{\trajProf}} \cdot \sum_{\trajProf \in \SC} \left ( p_\trajProf \cdot r_\trajProf \right)
= \sum_{\trajProf \in \SC} \left ( \frac{p_\trajProf}{\sum_{\trajProf \in \SC_{win}} p_{\trajProf}} \cdot r_\trajProf \right),
$$
and that
$$
\tauH^u_i = \frac{1}{ \sum_{\trajProf \in \SC_{win}}  \pH_{\trajProf} } \cdot \sum_{\trajProf \in \SC} \left( \pH_\trajProf \cdot r_\trajProf \right)
= \sum_{\trajProf \in \SC} \left( \frac{ \pH_\trajProf}{\sum_{\trajProf \in \SC_{win}}  \pH_{\trajProf} } \cdot r_\trajProf \right).
$$
Since each transtion probability is lower bounded by $1/P$, it is estimated to within a multiplicative error of $\left( 1 \pm O(P\epsilon) \right)$.
Since $p_{\trajProf}$ and $\pH_{\trajProf}$ can be written as the product of at most $D$ probabilities, each term $\frac{p_\trajProf}{\sum_{\trajProf \in \SC_{win}} p_{\trajProf}}$ is within a multiplicative error of $\left( 1 \pm O(D P \epsilon) \right)$ from $\frac{ \pH_\trajProf}{\sum_{\trajProf \in \SC_{win}}  \pH_{\trajProf} }$.
It follows that $\tau_{fair}$ is within a multiplicative factor of $\left( 1 \pm O(D P \epsilon) \right)$ from $\tauH^u_i$.
But notice that $\tauH^u_i \leq DB$, which implies that $\tau_{fair} \geq \tauH^u_i - O(D^2BP \cdot \epsilon) = \tauH^u_i - O(L \cdot \epsilon)$.
\end{proof}

\subsection{Designing an Adaptive Strategy for DAG-Utility Maximization}
\label{subsubsec:FrugalRobusttoMax}
From the previous section we know how to obtain close estimates of the grades.
Now we use well-estimated grades  to design a robust adaptive strategy for \DAG-\UtiMax and prove Theorem~\ref{thm:FrugalRobusttoAdaptiveMax}. Theorem~\ref{thm:FrugalRobusttoAdaptiveMax} directly follows by combining Lemma~\ref{lem:boundAdapMax} and the following Lemma~\ref{lem:FrugalRobustMax}.

\begin{lemma} \label{lem:FrugalRobustMax}
Assuming the conditions of Theorem~\ref{thm:FrugalRobusttoAdaptiveMax} and that the grade of each state is estimated to within an additive factor of $\epsilon / 4kD_i$, where $D_i$ is the depth of $\MS_i$, then there exists an adaptive algorithm \ALGH with utility at least
$$
\frac{1}{\alpha} \cdot \E_{\trajProf}\left[ \max_{\I \in \F}\{ \val(\I, \bYmax(\trajProf)) \} \right] - \epsilon.
$$
\end{lemma}

To prove Lemma~\ref{lem:FrugalRobustMax}, we  describe  our algorithm \ALGHA (Algorithm~\ref{alg:FrugalRobustToAdaptiveMax}).
We  define $\bYHmax$ as follows.

\begin{definition} \label{defn:SurrogateHatMax}
Fix a trajectory profile $\trajProf$ where each Markov system reaches the destination state.
For each $i$ and $u \in V_i$, let $d_u(\traj_i)$ be the number of transitions  for $\MS_i$ to reach $u$ from $s_i$ by taking the trajectory $\traj_i \in \trajProf$.
Let $\gamH^u_i(\traj_i) = \tauH^u_i + d_u(\traj_i) \epsilon/2kD_i$.
Define $\YHmax_{\traj_i} \overset{\Delta}{=} \min_{u \in \traj_i}\{ \gamH^u_i(\traj_i) \}$.
Denote the list of $\YHmax_{\traj_i}$'s as $\bYHmax(\trajProf)$ and $\YmHmax(\trajProf)$ the list of $\YHmax_{\traj_i}$ values in the set $M$.
\end{definition}

The key idea in $\ALGHA$ (the main difference from Algorithm~\ref{alg:frugalToAdaptiveMax}) is the ``upward shifting'' technique in Step~\ref{alg:computeViRobMax}.
As we advance a Markov system, we shift our estimates of its grades upward.
This guarantees that our algorithm is optimal in the teasing game \TG defined for Claim~\ref{claim:FairTeasingGameMax}.
\begin{algorithm}
\caption{Algorithm $\ALGHA$ for \UtiMax in \MPOI}
\label{alg:FrugalRobustToAdaptiveMax}
\begin{algorithmic}[1]
\State Start with $M=\emptyset$. Set $v_i=0$ and $\ctr_i=0$ for all elements $i$.
\State For each element $i\not\in M$, set $v_i=g\left(\YmHmax, i ,\tauH_i^u + \ctr_i \cdot \epsilon / 2kD_i \right)$ where $u$ is the current state of $i$.\label{alg:computeViRobMax}
\State Consider the element $j = \argmax_{i\not \in M~\&~ M\cup i \in \F} \{v_i\}$ and $v_j>0$.
\State Proceed $\MS_j$ for one step and set $\ctr_j=\ctr_j + 1$. If $t_j$ is reached by $\MS_j$, select $j$ into $M$.
\State If every element $i\not\in M$ has $v_i \leq 0$ then return set $M$. Else, go to Step~\ref{alg:computeViRobMax}.
\end{algorithmic}
\end{algorithm}

\begin{proof}[Proof of~Lemma~\ref{lem:FrugalRobustMax}]
This lemma immediately follows from the following two claims (whose proofs are in Appendix~\ref{sec:MissingProofRobustness}).

\begin{claim} \label{claim:POIDisutilityRobustMax}
The utility of running $\ALGHA$ in the real world is exactly the same as
$$
\E_{\trajProf} \left[ \val(Alg(\bYHmax(\trajProf), \A), \bYmax(\trajProf)) \right].
$$
\end{claim}

\begin{claim} \label{claim:RobustnessAlgMax}
For any trajectory profile $\trajProf$ and for any $i$, $|\YHmax_{\traj_i} - \Ymax_{\traj_i}| \leq \epsilon / 2k$.
Thus
$$
\val(Alg(\bYHmax(\trajProf), \A), \bYmax(\trajProf)) \geq \frac{1}{\alpha} \cdot \max_{\I \in \F}\{ \val(\I, \bYmax(\trajProf)) \} - \epsilon.
$$
\end{claim}

\eat{Before we prove the claims above, we first show how these claims imply Lemma~\ref{lem:FrugalRobustMax}.
We will also use $\ALGHA$ to denote the disutility of running $\ALGHA$ in the real world.
By Claim~\ref{claim:POIDisutilityRobustMax} and Claim~\ref{claim:RobustnessAlgMax}
\begin{eqnarray*}
&&\ALGHA = \E_{\trajProf} \left[ \val(Alg(\bYHmax(\trajProf), \A), \bYmax(\trajProf)) \right]\\
&& \geq \frac{1}{\alpha} \cdot \E_{\trajProf} \left[ \max_{\I \in \F}\{ \val(\I, \bYmax(\trajProf)) \} \right] - \epsilon
\end{eqnarray*}
which finishes the proof of the lemma.
Now the only thing left is to prove the three claims above.
\begin{proofof}{Claim~\ref{claim:SameSolutionRobustMax}}

\begin{claim} \label{claim:SameSolutionRobustMax}
For any trajectory profile $\trajProf$, running $\ALGHA$ in the real world returns the same solution as running $\A$ on $\bYHmax(\trajProf)$.
\end{claim}
The proof of this claim is the same as the proof of Claim~\ref{claim:sameSolutionMax} by thinking of $\gamH^u_i(\traj_i)$ as the grade of state $u$ in $\MS_i$ and $\YHmax_{\traj_i}$ as the quantity that defines the surrogate probelm, together with the observation that in Step~\ref{alg:computeViRobMax} of Algorithm~\ref{alg:FrugalRobustToAdaptiveMax}, $-(\tauH_i^u - \ctr_i \cdot \theta/D_i) = \gamH^u_i(\traj_i)$.
\end{proofof}
}

\end{proof}

%% file: promptness.tex

In \S\ref{sec:LPBoundOPT}, we give an LP relaxation to upper bound the optimum utility without the commitment constraint.
In \S\ref{sec:roundingLP}, we apply an OCRS to round the  LP solution to obtain an adaptive policy, while satisfying the commitment constraint.

\subsection{Upper Bounding the Optimum Utility} \label{sec:LPBoundOPT}

  Define the following variables, where $i$ is an index for the Markov systems.
  \begin{itemize} [noitemsep,topsep=5pt]
  \item
    $y_{i}^u$: probability we reach state~$u$ in Markov system~$\MS_i$ for  $u \in V_i \setminus T_i$.
  \item
    $z_{i}^u$: probability we play~$\MS_i$ when it is in state~$u$ for $u \in V_i \setminus T_i$.
  \item
    $x_i = \sum_{u \in T_i} z_{i}^u$: probability~$\MS_i$ is selected into the final solution when in a destination state.
    \item $P_{\F}$ is a convex relaxation containing all feasible solutions for packing  $\F$.
  \end{itemize}
We can now formulate the following LP, which is inspired from~\cite{GM-STOC07}.
  \begin{align*}
  \max_{\bz}    \qquad \qquad \sum_i &\Big( \sum_{u \in T_i} \reward_i^u z_{i}^u - \sum_{u \in  V_i \setminus T_i} \pi_{i}^u z_{i}^u \Big)\\
  \text{subject to} \qquad   \qquad  y_{i }^{s_i} &= 1 && \forall i\in J \\
    y_{i}^u &= \textstyle{\sum_{v\in V_i} (P_i)_{uv} z_{i}^v } && \forall i \in J, \forall u \in V_i \setminus s_i \\
    x_i &= \textstyle{\sum_{u \in T_i} z_{i}^u }&& \forall i\in J \\
    z_{i}^u &\leq  y_{i}^u && \forall i\in J, \forall u \in V_i\\
    \bx &\in P_{\F} \\
    x_i, y_{i}^u, z_{i}^u &\geq 0 && \forall i \in J, \forall u \in V_i
  \end{align*}
 The first four constraints characterize the dynamics in advancing the Markov systems.
The fifth constraint encodes the packing constraint $\F$.
We denote the optimal solution of this LP as $(\bx,\by,\bz)$. We can efficiently solve the above LP for packing constraints such as matroids, matchings, and intersection of $k$ matroids.

If we interpret the  variables $y_{i}^u, x_i$, and $z_{i}^u$ as the probabilities corresponding to the optimal strategy without commitment, it forms a feasible solution to the LP. This implies the following claim.
\begin{lemma} \label{lem:LPBoundOPT}
The optimum utility  without commitment is at most the LP value.
\end{lemma}

\subsection{Rounding the LP Using an OCRS} \label{sec:roundingLP}
Before describing our rounding algorithm, we define an OCRS. Intuitively, it is an online algorithm that  given a random set ground elements, selects a  feasible subset of them. Moreover, if it can guarantee that  every $i$ is selected w.p. at least $\frac{1}{\alpha} \cdot x_i$, it is called
 $\frac{1}{\alpha}$-selectable.

\begin{definition}[OCRS~\cite{FSZ-SODA16}] Given a point $x \in P_{\F}$, let $R(x)$ denote a random set containing each $i$ independently w.p. $x_i$. The elements $i $ reveal one-by-one whether  $i \in R(x)$ and we need to decide irrevocably whether to select an $i \in R(x)$ into the final solution before the next element is revealed. An OCRS is an online algorithm that selects a subset $I \subseteq R(x)$ such that $I \in \F$.
\end{definition}

\begin{definition}
[$\frac{1}{\alpha}$-Selectability~\cite{FSZ-SODA16}] Let $\alpha \geq 1$. An OCRS for $\F$ is $\frac{1}{\alpha}$-selectable if for any $x \in P_{\F}$ and all $i$, we have $\Pr[i \in I \mid i \in R(x)] \geq \frac{1}{\alpha} $.
\end{definition}

Our algorithm $\ALG$ uses  OCRS as an oracle.
It starts by fixing an arbitrary order $\pi$ of the Markov systems. (Our algorithm works even when an adversary decides the order of the Markov systems.)
Then at each step, the algorithm considers the next element $i$ in $\pi$ and queries the OCRS whether to select  element $i$ if it is \prepared.
If OCRS decides to select $i$, then \ALG advances the Markov system such that it plays from each state $u$ with independent probability $z_{i}^u/y_{i}^u$.
This guarantees that the desination state is reached with probability $x_i$.
If OCRS is not going to select $i$, then \ALG moves on to the next element in $\pi$.
A formal description of the algorithm can be found in Algorithm~\ref{alg:OCRStoCommitment}.

\begin{algorithm}
\caption{Algorithm $\ALG$ for Handling the Commitment Constraint}
\label{alg:OCRStoCommitment}
\begin{algorithmic}[1]
\State Fix an arbitrary order $\pi$ of the items. Set $M=\emptyset$ and pass $\bx$ to OCRS.
\State Consider the next element $i$ in the order of $\pi$. Query OCRS whether to add $i$ to $M$ if $i$ is \prepared.  \label{alg:ConsiderNextEle}
\Statex (a) If OCRS would add $i$ to $M$, then keep advancing the Markov system: play from each current state $u \in V_i\setminus T_i$ independently w.p. $z_{i}^u/y_{i}^u$, and otherwise go to Step~\ref{alg:ConsiderNextEle}. If a destination state $t$ is reached then add $i$ to $M$ w.p. $z_{i}^t/y_{i}^t$.  \label{alg:AdvanceChain}
\Statex (b) Go to Step~\ref{alg:ConsiderNextEle}.
\end{algorithmic}
\end{algorithm}

We show below that \ALG has a utility of at least $1/\alpha$ times the LP value.
\begin{lemma} \label{lem:ConstantTimesLPOpt}
The utility of \ALG is at least $1/\alpha$ times the LP optimum.
\end{lemma}
Since by Lemma~\ref{lem:LPBoundOPT} the LP optimum is an upper bound on the utility of any policy without  commitment, this proves Theorem~\ref{thm:OCRStoCommitment}.
We now prove Lemma~\ref{lem:ConstantTimesLPOpt}.
\begin{proof}[Proof of Lemma~\ref{lem:ConstantTimesLPOpt}]
Recollect that we  call a Markov system \prepared if it reaches an absorbing destination state.
We first notice that once \ALG starts to advance a Markov system $i$, then by Step~\ref{alg:AdvanceChain} of Algorithm~\ref{alg:OCRStoCommitment}, element $i$ is \prepared with probability exactly $x_i$. This agrees with what \ALG tells the OCRS.
Since the OCRS is $1/\alpha$-selectable, the probability that any Markov system $\MS_i$ begins advancing is  $1/\alpha$.
Here the probability is both over the random choice of the OCRS and the randomness due to the Markov systems. Conditioning on the event that $\MS_i$ begins advancing, the probability that it is selected into the final solution on reaching a  destination state $t \in T_i$ is exactly $z_{i}^t$.
Hence, the conditioned utility from Markov system $\MS_i$ is exactly
\[
\textstyle{\sum_{u \in T_i} \reward_i^u z_{i}^u - \sum_{u \in  V_i \setminus T_i} \pi_{i}^u z_{i}^u.}
\]
By removing the conditioning and by linearity of expectation, the utility of \ALG is at least
$
\frac{1}{\alpha} \cdot  \sum_i \Big( \sum_{u \in T_i} \reward_i^u z_{i}^u - \sum_{u \not\in T_i} \pi_{i}^u z_{i}^u \Big),
$
which proves this lemma.
\end{proof}

%% file: relatedwork.tex
The field of combinatorial optimization has been extensively studied: we refer the readers to Schrijver's popular book~\cite{Schrijver-Book03}, and the references therein. In recent years,
there has also been a lot of interest in  studying these combinatorial problems  for stochastic inputs.
\cite{DGV-FOCS04,DGV05,GM-SODA07TALG12,GM-STOC07,BGK-SODA11,LiYuan-STOC13,Ma-SODA14}  considered  stochastic knapsack, \cite{CIKMR09,A11,BGLMNR-Algorithmica12,BCNSX15,AGM15} studied stochastic matchings, \cite{GuhaM09,GKNR-SODA12,BN-IPCO14} studied stochastic orienteering,  \cite{ANS-WINE08,GN-IPCO13,ASW14,GNS-SODA17,GNS-SODA16} considered stochastic submodular maximization, and \cite{GM-STOC07,GuhaM09,GKMR-FOCS11,Ma-SODA14} studied budgeted multi-armed bandits.  These works (besides~\cite{GM-STOC07}) do not consider mixed-sign utility objective or multi-stage probing, which is our primary focus.

%% file: appendix.tex

\appendix



\section{Proof of Lemma~\ref{lem:boundAdapMax}}
\label{misproof:boundAdapMax}
We restate Lemma~\ref{lem:boundAdapMax} below.

\boundAdapMax*

\begin{proof}
We abuse the notation and use \OPT to  denote both the optimal policy and its utility. Suppose we fix a trajectory profile $\trajProf$ where each Markov system $\MS_i$ reaches a destination state. Let $\I(\trajProf)$ be the set of elements selected by \OPT on $\trajProf$, where  notice that some of the unselected elements may not be \prepared:
\OPT might   have  selected $\I(\trajProf)$ only after playing prefixes of  trajectories in $\trajProf$.
The following observation follows  from the definition of $\Sur(\trajProf)$.
\begin{observation} \label{obs:OptToSur}
For any trajectory profile $\trajProf$,
$$
\val(\I(\trajProf),\bYmax(\trajProf))\leq \Sur(\trajProf).
$$
\end{observation}

\noindent Now,  using the following  Lemma~\ref{claim:OptAgainstOptYMax} along with Observation~\ref{obs:OptToSur} finishes the proof of Lemma~\ref{lem:boundAdapMax}.
\begin{lemma} The utility of the optimal strategy
\label{claim:OptAgainstOptYMax}
\[ \OPT \leq \E_{\trajProf} \left[\val(\I(\trajProf),\Ymax(\trajProf))\right] .  \]
\end{lemma}
\end{proof}

\begin{proof}[Proof of Lemma~\ref{claim:OptAgainstOptYMax}]
Since for every trajectory profile $\trajProf$ both $\OPT$ in the \MPOI
world and $\E_{\trajProf} \left[\val(\I(\trajProf),\Ymax(\trajProf))\right] $ in the
\FreeInfo world pick the same set of elements $\I(\trajProf)$,  the expected
value due to the set function $h$ is the same. Hence, WLOG assume $h(\I)=0$ for all $\I\in \F$.

Now  consider the following \emph{teasing game} $\TG$ defined using the prevailing cost from Definition~\ref{defn:PrevailingCostMax}.
Consider a game where each Markov system $\MS_i$ starts at its initial state $s_i$ and a player is invited to advance  the Markov systems.
Besides advancing, the player is allowed to select any arbitrary elements (need not be feasible in $\F$) or terminate the game at any time during the game.
Whenever an element $i$ is selected, the player  pays a corresponding \emph{cost}, which is  set to be the prevailing cost defined by the trajectory that lead to the current state in  $\MS_i$. The player's goal is to maximize the expected \emph{value}, which is the expected utility (as defined for \UtiMax) from advancing the Markov systems \emph{minus} the expected total cost he pays when some items are selected. Observe that in this game the costs are updated in a ``teasing'' manner according to the prevailing costs that motivates the player to continue playing.
By an argument similar to~\cite{DTW-SIDMA03}, we have the following lemma.
\begin{lemma}
\label{claim:FairTeasingGameMax}
The teasing game \TG is \emph{fair}, which means that no strategy achieves a positive expected value by playing it and that there exists  a strategy with  zero expected value. Moreover, the following  strategy plays fairly: irrespective of the  order in which the Markov systems are played, whenever the player starts to advance a Markov system, he continues to advance it through the entire epoch.
\end{lemma}

Now consider running the optimal policy $\OPT$ in the teasing game. Let $\trajProf$ be a trajectory profile in which each chain reaches its destination state.
Let $\trajProf_T$ denote a trajectory profile  until the moment when $\OPT$ returns the solution $\I(\trajProf)$ on the trajectory profile $\trajProf$.
It should be noticed that each trajectory in $\trajProf_T$ is a prefix of the corresponding trajectory in $\trajProf$.
In particular, for an element $i \in \I(\trajProf)$, $\traj_i$ coincides with $(\trajProf_T)_i$ since the destination state of $\MS_i$ is reached.
For an element $i \notin \I(\trajProf)$, however, $(\trajProf_T)_i$ may only be a prefix of $\traj_i$.
It follows that applying $\OPT$ in $\TG$ along trajectory profile $\trajProf$ incurs a cost of $
\sum_{i \in \I(\trajProf)} \Ymax_{(\trajProf_T)_i}$,
where $\Ymax_{(\trajProf_T)_i}$ is the prevailing cost for $\MS_i$ on trajectory $(\trajProf_T)_i$ according to Definition~\ref{defn:PrevailingCostMax}.
Since $\TG$ is a fair game, the expected utility of $\OPT$ cannot be larger than the expected cost it pays, i.e.,
\[
 \OPT \leq \E_{\trajProf}\Big[ \sum_{i \in \I(\trajProf)} \Ymax_{(\trajProf_T)_i} \Big].
\]
Since the elements  $i \in \I(\trajProf)$ are \prepared, we have  $\traj_i = (\trajProf_T)_i$ and
 $$
\sum_{i \in \I(\trajProf)} \Ymax_{(\trajProf_T)_i} = \sum_{i \in \I(\trajProf)} \Ymax_{\traj_i}.
 $$
This implies
 $$
 \OPT \leq  \E_{\trajProf}\Big[ \sum_{i \in \I(\trajProf)} \Ymax_{\traj_i} \Big],
 $$
which finishes the proof of Lemma~\ref{claim:OptAgainstOptYMax}.
\end{proof}


\section{Proof of Lemma~\ref{lem:convertFrugalAlgMax}}
\label{misproof:convertFrugalAlgMax}
We restate Lemma~\ref{lem:convertFrugalAlgMax} below.

\convertFrugalAlgMax*

\begin{proof}[Proof of Lemma~\ref{lem:convertFrugalAlgMax}]
  We  describe how to adapt the \frugal Algorithm~\A to an adaptive strategy $\ALGA$ in the \MPOI world. $\ALGA$ uses the grade $\tau$ as proxy for $\bYmax$, since $\bYmax$ is known only when the Markov systems reach their destination states. More specifically, at each moment when the \frugal Algorithm~\A is trying to evaluate the marginal-value function for each element, instead of using the $\bYmax$ value for each element, which we may not yet know at the moment, the strategy uses the $\tau$ values to compute the marginal. For the element chosen by~\A, the corresponding Markov system will be advanced one more step. A more specific description of our algorithm $\ALGA$ is given Algorithm~\ref{alg:frugalToAdaptiveMax}.
Here $\Ymmax$ for a set $M \subseteq J$ is defined as the list of $\bYmax$ values that are in the set $M$.

\setlength{\intextsep}{5pt}
\begin{algorithm}
\caption{\ALGA for \UtiMax in \MPOI}
\label{alg:frugalToAdaptiveMax}
\begin{algorithmic}[1]
\State Start with $M=\emptyset$ and  $v_i=0$ for all elements $i$.
\State For each element $i\not\in M$, set $g(\Ymmax , i, \tau^{u_i}_i)$ where $u_i$ is the current state of $i$.\label{alg:computeViMax}
\State Consider the element $j = \argmax_{i\not \in M~\&~ M\cup i \in \F} \{v_i\}$.
\State If $v_j>0$, then if $\MS_j$ is not in a destination state then proceed $\MS_j$ by one step and go to Step~\ref{alg:computeViMax}.
Else, when $v_j>0$ but  $\MS_j$ is  in a destination state $t_j$, select $j$ into $M$ and go to Step~\ref{alg:computeViMax}.
\State Else, if every element $i\not\in M$ has $v_i \leq 0$ then return set $M$.
\end{algorithmic}
\end{algorithm}


In the following Claim~\ref{claim:sameSolutionMax}, we argue that for any trajectory profile $\trajProf$, running $\ALGA$ in \MPOI returns the same set of elements as running $\A$ for $\bYmax(\trajProf)$.
\begin{claim}[Claim~\ref{claim:sameSolutionMax}]
\label{claim:sameSolutionMax}
For any trajectory profile $\trajProf$,
the solution returned by running Algorithm~\ref{alg:frugalToAdaptiveMax} in the \MPOI world is the same as   the solution by  Algorithm~\A on $\bYmax(\trajProf)$.
\end{claim}
Before proving Claim~\ref{claim:sameSolutionMax}, we use it to prove Lemma~\ref{lem:convertFrugalAlgMax} by showing that the utility of Algorithm~\ref{alg:frugalToAdaptiveMax} in the \MPOI world is at least
\[\E_{\trajProf} [\val(\A(\bYmax(\trajProf)), \bYmax(\trajProf))].
\]

By Claim~\ref{claim:sameSolutionMax}, the value due to the set function $h$ is the same for both algorithms. So without loss of generality, assume $h$ is always 0. We consider the teasing game $\TG$ as defined in Claim~\ref{claim:FairTeasingGameMax}.
By definition, $g$ is an increasing function of the last parameter $y$.
Since grade is used as that parameter and the grade of each state visited during an epoch is at least the grade of the initial state of that epoch,
 it follows that once Algorithm~\ref{alg:frugalToAdaptiveMax} starts to play a Markov system $\MS_i$, it will not switch  before finishing an epoch. Therefore, by Claim~\ref{claim:FairTeasingGameMax}, Algorithm~\ref{alg:frugalToAdaptiveMax} plays a fair game. So the expected cost that Algorithm~\ref{alg:frugalToAdaptiveMax} pays is the same as its expected utility from playing the Markov systems. However,  Claim~\ref{claim:sameSolutionMax} gives the expected cost payed by Algorithm~\ref{alg:frugalToAdaptiveMax} is the same as the utility of running Algorithm~\A in the \FreeInfo world, i.e., $\E_{\trajProf} [\val(\A(\bYmax(\trajProf)), \bYmax(\trajProf))]$. Hence, the utility of running Algorithm~\ref{alg:frugalToAdaptiveMax} is at least $\E_{\trajProf} [\val(\A(\bYmax(\trajProf)), \bYmax(\trajProf))]$.
\end{proof}


It remains to prove the missing Claim~\ref{claim:sameSolutionMax} in the proof of Lemma~\ref{lem:convertFrugalAlgMax}.
\begin{proof}[Proof of Claim~\ref{claim:sameSolutionMax}]
Suppose we fix a trajectory profile $\trajProf$ where each Markov system reaches some destination state. We prove the claim by induction on the number of elements already selected into the set $M$. Suppose the set of elements selected into $M$ is the same by running the two algorithms until now. We  show that the next element selected by the algorithms into $M$ is the same.

Assume for the purpose of contradiction that the next element picked by \A is $j$ but the next element picked by Algorithm~\ref{alg:frugalToAdaptiveMax} is $i\neq j$. By the definition of Algorithm~\A,
\begin{align}
\label{eqn:DefnOfJ}
j=\argmax_{i' \notin M} \left\{ g \left( \Ymmax(\trajProf), i',\Ymax_{\traj_{i'}} \right) \right\}.
\end{align}
where $\traj_{i}'$ denotes the trajectory of $\MS_{i'}$ in $\trajProf$.
Now we look at the trajectory $\traj_i$, it follows that the prevailing cost $\Ymax_{\traj_i}$ is non-increasing over this trajectory and is equal to $\Ymax_{\traj_i}$ when $\MS_i$ reaches the destination state. We look at the last moment $t_0$ when the prevailing cost of $\MS_i$ decreases. Consider the first moment $t_1$ after $t_0$ that our Algorithm~\ref{alg:frugalToAdaptiveMax} decides to play $\MS_i$ (but has not actually played $\MS_i$ yet). It follows that the prevailing cost of $\MS_i$ at moment $t_1$ is exactly the same as $\Ymax_{\traj_i}$ and also the grade $\tau_i^{u_i}$ of the current state $u_i$. Denote $\Ymax_{\traj'_j}$ the prevailing cost of $\MS_j$ and $u_j$ the state of $\MS_j$ at moment $t_1$. Then we have $\Ymax_{\traj'_j} \geq \Ymax_{\traj_j}$ because the prevailing cost of $\MS_j$ is also non-increasing. By the definition of $t_1$, one has
\begin{align*}
g \left( \Ymmax(\trajProf), i,\Ymax_{\traj_i} \right) &= g \left( \Ymmax(\trajProf), i, \tau_i^{u_i} \right) \\
&> g \left( \Ymmax(\trajProf), j, \tau_j^{u_j} \right) \geq g \left( \Ymmax(\trajProf), j,\Ymax_{\traj'_j} \right).
\end{align*}
However, since $g$ is increasing in the last parameter, it follows that
$$g \left( \Ymmax(\trajProf), j,\Ymax_{\traj'_j} \right) \geq g \left( \Ymmax(\trajProf), j,\Ymax_{\traj_j} \right),$$
which implies
\begin{gather*}
  g \left( \Ymmax(\trajProf), i,\Ymax_{\traj_i} \right) > g \left(
    \Ymmax(\trajProf), j,\Ymax_{\traj_j} \right).
\end{gather*}
This contradicts with the definition of $j$ in Eq~\eqref{eqn:DefnOfJ}.
\end{proof}


\section{Comparing Grade and Weitzman's Index for Pandora's Box} \label{sec:GradeVsWeitzman}
Recall Weitzman's Pandora's box formulation of the oil-drilling problem mentioned in Section~\ref{sec:intro}. Given probability distributions of $n$ independent random variables $X_i$ (amount of oil at site $i$) and their \emph{probing} (inspection) prices $\price_i$, the goal is to design a strategy to \emph{adaptively} probe a set $\probed$ to maximize expected utility
\[ {\E \Big[ \max_{i\in \probed} \{X_i\} - \sum_{i\in \probed} \price_i  \Big]. }\]

The Weitzman's index for site $i$, denoted by $\taumax_i$, is defined using the following equation $\E[(X_i - \taumax_i)^+] = \pi_i$.
It is known that the following strategy is optimal~\cite{Weitzman-Econ79}.

\emph{Selection Rule:} The next site to be probed is the one with with the highest Weitzman's index.

\emph{Stopping Rule:} Terminate when the maximum realized value amongst the probed sites exceeds the Weitzman's index of every unprobed site.

It turns out that Weitman's index $\taumax_i$ is simply the grade, defined in Section~\ref{para:grade}, in disguise.
To see this, we start by noticing that each variable $X_i$ with probing price $\price_i$ can be thought of as the following Markov system.
There is one initial state $s_i$ with moving cost $\price_i$.
$s_i$ has transitions, with probabilities according to the distribution of $X_i$, to a set $T_i$ of destination states, each corresponding to a possible outcome of the variable $X_i$ .
The value of each destination state is naturally set to be the corresponding outcome of $X_i$.
We show below that $\taumax_i$ is simply the grade $\tau_i^{s_i}$ of the initial state $s_i$.

According to our definition of grade in Section~\ref{para:grade}, in the $\tau_i^{s_i}$-penalized Markov game $\MS(\tau_i^{s_i})$, 
there is a fair strategy that probes site $i$ and achieves a zero  utility.
Such a strategy would pick site $i$ (i.e., play in the corresponding destination state) if and only if $X_i - \tau_i^{s_i} \geq 0$.
The utility of that policy is thus $-\price_i + \E[(X_i - \tau_i^{s_i})^+] = 0$.
Comparing with the definition of Weitzman's index, this shows $\taumax_i = \tau_i^{s_i}$.
The optimality of Weitzman's strategy is therefore also implied by Theorem~\ref{thm:frugalToAdaptiveMax}.

\section{Adaptive Algorithms for Disutility Minimization} \label{sec:DisutilityMinimization}

We give the corresponding definitions for the \DisMin problem.
\begin{definition}[Prevailing Reward for \DisMin]
\label{defn:PrevailingRewardMin}
The  \emph{prevailing reward} of $\MS_i$ for the trajectory $P_i$ in \DisMin is defined as
$$
\Rmin_{P_i} \overset{\Delta}{=} \max_{u \in P_i}\{ -\tau_i^{u} \}.
$$
For a trajectory profile $\trajProf$, denote $\Rmin_\trajProf$ the list of prevailing rewards for each Markov system.
\end{definition}

For a trajectory $P_i$ in the \DisMin problem, consider the change of the prevailing reward as the Markov system starts from $s_i$ and moves according to $P_i$.
It follows that the prevailing reward is non-decreasing in this process.
Moreover, it increases whenever the Markov system reaches a state that has smaller grade than each previously visited state.
Now we are ready to state the definition of an \emph{epoch}.
\begin{definition}[Epoch for \DisMin]
\label{defn:EpochMin}
An \emph{epoch} is defined to be the period from the time when the prevailng reward increases until the moment just before the next time it increases.
\end{definition}
It follows that within an epoch, all states visited has grade no smaller than the prevailing reward at the start of this epoch and thus the prevailing reward stays constant in an epoch.
We can therefore view the prevailing reward as a non-decreasing piece-wise constant function of time.

\begin{definition}[\frugal Covering  Algorithm]\label{defn:frugalCovering}
For a \DisMin problem in the \Det world with covering constraints $\F$ and cost function $\cost$, we say Algorithm \A is \frugal if  there exists a \emph{marginal-value} function $g(\Y,i,y):\R^J \times J \times \R \rightarrow \R$ that is decreasing in $y$, and  for which  the pseudocode is given by Algorithm~\ref{alg:frugalCoverageMin}.
Moreover, the function $g(\Y,i,y)$ should \emph{encode} the constraints $\F$, such that whenever $M$ is infeasible, then $\exists i \notin M$ with $v_i>0$. This requirement will ensure that a feasible solution is returned.
\begin{algorithm}
\caption{\frugal Covering Algorithm \A}
\label{alg:frugalCoverageMin}
\begin{algorithmic}[1]
\State Start with $M=\emptyset$ and  $v_i=0$ for each element $i \in J$.
\State For each element $i\not\in M$, compute  $v_i = g( \Ym, i,Y_i)$. Let $j = \argmax_{i\not \in M} \{v_i\}$.\label{alg:FrugalcomputeViCoverage}
\State If $v_j>0$ then add $j$ into $M$ and  go to Step~\ref{alg:FrugalcomputeViCoverage}. Otherwise, return $M$.
\end{algorithmic}
\end{algorithm}
\end{definition}

With the definitions above, one can prove the following theorem for \DisMin using similar techniques as in Section~\ref{sec:UtilityMaximization}.

\begin{restatable}{theorem}{frugalToAdaptiveMin} \label{thm:frugalToAdaptiveMin}
For a semiadditive objective function $\cost$, if there exists an  $\alpha$-approximation \frugal algorithm for a \DisMin problem over some covering constraints $\F$ in the \FreeInfo world, then there exists an $\alpha$-approximation strategy for the corresponding \DisMin problem in the \MPOI world.
\end{restatable}

\section{Missing Proofs in the Robustness Model}
\label{sec:MissingProofRobustness}
\begin{proofof}{Claim~\ref{claim:POIDisutilityRobustMax}}
Because $\ALGHA$ shifts the estimated grade upward by $\epsilon/2kD_i$ each time we advance $\MS_i$ and that each grade is estimated to within an additive error of $\epsilon/2kD_i$, whenever $\ALGHA$ starts to advance a Markov system, it continues to advance it through the whole epoch.
It follows from Claim~\ref{claim:FairTeasingGameMax} that $\ALGHA$ is an optimal policy in the teasing game $\TG$.
By a similar argument as the proof of Claim~\ref{claim:sameSolutionMax}, one can show that for any list of trajectories $\trajProf$, running $\ALGHA$ in the real world returns the same solution as running $\A$ on $\bYHmax(\trajProf)$.
These imply the claim.
\end{proofof}

\begin{proofof}{Claim~\ref{claim:RobustnessAlgMax}}
Since Markov system $i$ can be played at most $D_i$ times, it follows that the estimated grade is shifted upward by at most $(D_i-1)\epsilon/2kD_i$.
It follows that each estimated grade after the upward shifting is still within an additive error of $\epsilon/2k$ from the real grade, which finishes the first part of the grade.

The second part follows from the following inequalities.
\begin{eqnarray*}
&&\val(Alg(\bYHmax(\trajProf), \A), \bYmax(\trajProf)) \\
&&\geq \val(Alg(\bYHmax(\trajProf), \A), \bYHmax(\trajProf)) - k \cdot \epsilon/2k \\
&&\geq \frac{1}{\alpha} \cdot \max_{\I \in \F} \left \{ \val ( \I , \bYHmax(\trajProf)) \right \} - \epsilon/2\\
&&\geq \frac{1}{\alpha} \cdot \val \left(\argmax_{\I \in \F} \left \{ \val ( \I ,\bYmax(\trajProf)) \right\}, \bYHmax(\trajProf) \right) - \epsilon/2\\
&&\geq \frac{1}{\alpha} \cdot \max_{\I \in \F} \left \{ \val ( \I , \bYmax(\trajProf)) \right \} - \epsilon,
\end{eqnarray*}
where the last line follows because $\alpha \geq 1$.
\end{proofof}


\section{Assumptions in the Robustness Model}
\subsection{DAG Assumption}
\label{sec:DAGNecessaryforRobust}
We give an example to illustrate why the \DAG assumption is necessary for our robustness results to hold.
We show that if there are cycles in the Markov chains, one might need to estimate the input parameters to a super-exponentially accurate precision in order to achieve a small additive loss in the performance.

Consider the following \UtiMax problem of picking at most one item (i.e. the constraint $\F$ is the uniform Matroid with rank 1) where \emph{all the input parameters are polynomially bounded}.
We have $n$ Markov systems $\{\MS_i\}_{1 \leq i \leq n}$.
The last $n-2$ Markov systems each has only one state, which is a destination state, with value 0.
These Markov systems can be safely ignored since one can pick nothing and obtains 0 utility.
We can therefore focus only  on the other two Markov systems.

The 2nd Markov system $\MS_2$ has only one state, which is a destination state, with value $1$.
The first Markov system $\MS_1$ has three states $\{s_1,v,t_1\}$, where $s_1$ is the initial state with playing cost $n^2/2^{2^n}$, $t_i$ is the destination state with value $n^2/2$, and $v$ is some intermediate state with playing cost 0.
The transitions in $\MS_1$ are as follows.
$s_1$ goes to $v$ deterministically.
$v$ goes to $s_1$ with probability $1-1/p2^{2^n}$ and $t_1$ with probability $1/p2^{2^n}$, where $p \in (0,1]$.
Notice that $\MS_1$ contains a cycle and a negligible transition out of the cycle to the destination.
It follows that the utility obtained by always playing $\MS_1$ is $n^2/2 - pn^2$, which is $n^2/4$ if $p=1/4$ and $-n^2/2$ if $p=1$.

In this case, if we fail to estimate the transition probabilities of $\MS_1$ to a super-exponentially accurate precision of $O(1/2^{2^n})$, it would render it impossible even to distinguish between the case where playing $\MS_1$ has utility $\Theta(n^2)$ and the case where playing $\MS_1$ has negative utility, which makes it impossible to obtain an approximation policy within a small additive error from the optimal policy.

\subsection{Polynomial Upper Bound on Input Parameters}
\label{sec:BoundedInput}
Here, we give an example to illustrate why Assumption~\ref{assump:BoundedInput} is necessary for our robustness results to hold. We show that if some parameters are exponential in the input parameter, then one might need to estimate some input parameters to within an additive error that is exponential in the input parameters.

Consider the following \UtiMax problem of picking at most one item (i.e. the constraint $\F$ is the uniform Matroid with rank 1) where \emph{all the input parameters are polynomially bounded}.
We have $n$ Markov systems $\{\MS_i\}_{1 \leq i \leq n}$. The last $n-1$ Markov systems deterministically give 0 utility.
The first Markov system $\MS_1$ has an initial state $s_1$ and two destination states $t_1$ and $t_2$.
The initial state $s_1$ has price $3^{n}$.
It goes to $t_1$ with probability $p$ and $t_2$ with probability $1-p$.
$t_1$ has reward $2\times 3^{n}$ and $t_2$ has reward $0$.

The player has to decide between playing $\MS_1$ or doing nothing at all. If $p = 1/2 + \Theta(1/2^{n})$, then the utility of playing $\MS_1$ is $\Theta(1.5^n)$ and if $p=1/2 - \Theta(1/2^n)$, then the utility of playing $\MS_1$ is $ - \Theta(1.5^n)$. It follows that one need to estimate the transition probabilities to within an additive error that is exponentially small.

\subsection{Other Assumptions Without Loss of Generality}
\label{sec:JustifyAssumptionsRobust}
Recall that for the \DAG-\UtiMax problem in the robustness model, we made the following assumptions.
\begin{itemize}
\item All non-zero transition probabilities are lower bounded by $1/P$, where $P$ is some polynomial in the parameters above.
\item We can estimate the prices $\bprice$ and the rewards $\breward$ exactly, i.e. the only unknown input parameters are the transition probabilities.
\end{itemize}

The assumption that all non-zero transition probabilities are polynomially lower bounded is without loss of generality.
It can be removed by the following procedure.
We start by setting a threshold $1/P$ and estimating all the data to within an additive error smaller than $1/P$.
We then ignore the transitions that have estimated probabilities smaller than $2/P$.
This is done by reallocating these probability masses to other transitions from the same state in both the original Markov systems and the estimated Markov systems.
After the removal of these negligible transition probabilities, the remaining Markov systems have a lower bound of $1/P$ on all the transition probabilities.
Since the maximum price paid on any sample path in a Markov system is at most $DB$, it follows that this changes the optimal policy by at most a very small additive factor if the polynomial $P$ we take is large enough. Therefore, we shall assume without loss of generality a lower bound on all non-zero transition probabilities.

The assumption that we can estimate the prices $\bprice$ and the rewards $\breward$ exactly is again without loss of generality and can be removed by the following argument with a small additive term in the theoretical guarantee.
Suppose all the prices $\bprice$ and the rewards $\breward$ are estimated within an additive error of $\delta/nD$.
Since one needs at most $D$ steps to reach the destination for each Markov system, the utility is affected by at most a small additive factor of $\delta/nD \times nD = \delta$ if we set $\delta$ to be small. Therefore, we will assume that estimations of the prices $\bprice$ and the rewards $\breward$ are exact and only the estimations of transition probabilities have deviations from the real transition probabilities.